\documentclass[11pt]{article}
\usepackage{graphicx} % Required for inserting images
\usepackage[T1]{fontenc}
\usepackage[utf8]{inputenc}
\usepackage[margin=1in]{geometry}

\usepackage{macros}
\usepackage{xcolor}
\usepackage{hyperref}
\usepackage{booktabs}
\usepackage{fullpage}

\crefname{maintheorem}{Theorem}{Theorems}
\Crefname{maintheorem}{Theorem}{Theorems}
\crefname{fact}{Fact}{Facts}
\Crefname{fact}{Fact}{Facts}

\hypersetup{colorlinks,linkcolor=cyan!50!black,citecolor=cyan!50!black,menucolor=cyan!50!black,filecolor=blue,urlcolor=blue}

\newtheorem{maintheorem}{Theorem}
\newtheorem{lemma}[maintheorem]{Lemma}
\newtheorem{corollary}[maintheorem]{Corollary}
\theoremstyle{definition}
\newtheorem{definition}[maintheorem]{Definition}

\title{Improved Subexponential Upper Bounds for $3$-Restricted Matching Vector Families}
 \author{
 	Sidhant Saraogi%
 	\thanks{Georgetown University. Email: \texttt{ss4456@georgetown.edu}.}
 }
\date{}

\begin{document}
\maketitle

\begin{abstract}
Matching Vector families (MVFs) are defined by two ordered lists of vectors in $\Z_m^n$ whose inner products satisfy specific residue patterns modulo an integer $m$. Most famously, restricted MVFs are used to construct the best-known constant-query Locally Decodable codes (LDCs). 

We prove an upper bound of $2^{O\left(\sqrt{n\log n \log m}\right)}$ on the size of $3$-restricted MVFs in $\Z_m^n$ for $m \leq \sqrt{n}$, substantially improving on the previous best bound of $2^{O(n/\log n)}$ by Bhowmick, Dvir and Lovett (STOC'13, SICOMP'14). Our proof relies on a new polynomial-method argument that controls collisions in sumsets of matching vectors. 
\end{abstract}
\section{Introduction}

A Matching Vector (MV) family is defined by two ordered lists of vectors $\U = \{u_1, \dots, u_K\}$ and $\V = \{v_1, \dots, v_K\}$ in $\Z_m^n$ for some $m, n \geq 1$ such that the inner products of pairs of vectors from $\cU$ and $\cV$ satisfy certain residue patterns modulo $m$. More precisely:

\begin{definition}[Matching Vector Family] \label{def:mvf}
Let $S \subseteq \Z_m \setminus \{0\}$. An $S$-matching vector family is a set of vectors $\U = \{u_1, \dots, u_K\}$, $\V = \{v_1, \dots, v_K\}$, $\U, \V \subseteq \Z_m^n$ such that:
\begin{itemize}
    \item $\forall i \in [K], \langle u_i, v_i \rangle = 0$.
    \item $\forall i \neq j$, $i, j \in [K]$, $\langle u_i, v_j \rangle \in S$.
\end{itemize}
\end{definition}

Such an MV family is referred to as a $|S|+1$-restricted MV family. When $S = \Z_m \setminus \{0\}$, we simply refer to it as an MV family. We refer to $K$ as the size of the MV family.  

Famously, MV families were used to construct the best-known constant-query Locally Decodable codes (LDCs) \cite{Yek08, Efr09, dvir2011matching}. MVFs have since found a wider array of applications including to private information retrieval \cite{Yek08, Efr09, GKS24, alon2025}, conditional disclosure of secrets \cite{alon2025}, circuit complexity, and catalytic computation \cite{HPR26}. They are also connected to the study of Ramsey graphs \cite{Grolmusz2000} and polynomials weakly representing the OR function \cite{Grolmusz}.

As is common, we denote by $\MV(m, n)$ the size of the largest possible MV family in $\Z_m^n$ and by  $\MV(m, n, r)$ the size of the largest possible $r$-restricted MV family in $\Z_m^n$. 

When $m$ has several distinct prime factors, the situation is much richer.
For suitable constant moduli $m$ with $t$ distinct prime factors, Grolmusz
\cite{Grolmusz2000} constructed MVFs whose off-diagonal inner products lie in
a canonical set of $2^t-1$ nonzero residues. In our terminology, these are
$2^t$-restricted MVFs, and they satisfy
\[
\MV(m,n,2^t)
\ge
\exp\left(\Omega\left(
\frac{(\log n)^t}{(\log\log n)^{t-1}}
\right)\right).
\]

This remains the best-known construction of MVFs. Based on \cite{bdl13} and the recently proved Polynomial Freiman-Rusza Theorem \cite{PFR24}, we know that the best upper bound for MVFs is 

\[
\MV(m, n) \leq (c_m)^{n /\log n}
\]
for some constant $c_m$ depending only on $m$.

\paragraph{Locally Decodable Codes (LDCs).}
The most surprising application of MVFs is to the construction of LDCs, codes that allow for extremely efficient decoding of individual message bits. 
\begin{definition}[Locally Decodable Codes]\label{def:ldc}
A $q$-ary code $C:\F_q^K\to\F_q^N$ is said to be $(r,\delta,\eps)$-locally decodable if there exists a randomized decoding algorithm $\mathcal A$ such that:
\begin{enumerate}
    \item For all $\mathbf{x}\in\F_q^K$, $i\in[K]$, and all vectors $\mathbf{y}\in\F_q^N$ such that
    \[
    d(C(\mathbf{x}),\mathbf{y})\le \delta N,
    \]
    we have
    \[
    \Pr[\mathcal A^{\mathbf{y}}(i)=\mathbf{x}(i)]\ge 1-\eps,
    \]
    where the probability is over the random coin tosses of $\mathcal A$.
    \item $\mathcal A$ makes at most $r$ queries to $\mathbf{y}$.
\end{enumerate}
\end{definition}

\cite{Yek08} showed that an $r$-restricted MV family of size $K$ in $\Z_m^n$ can be used to construct an $r$-query locally decodable code with message length $K$ and codeword length $m^n$. MVFs are central to the construction of matching vector codes. Building on this
framework, a line of works \cite{Yek08, Efr09, dvir2011matching} constructed the best-known constant-query
LDCs. In particular, the \cite{Efr09} construction gives $3$-query LDCs of codeword
length
\[
N=\exp\left(\exp\left(O\left(\sqrt{\log K\log\log K}\right)\right)\right).
\]
On the other hand, despite a long line of works \cite{KT00, KW04, GKST06, AGKM23, JM25, BHKL25}, we only know that 
\[
N \geq \tilde{\Omega}(K^{r/(r-2)})
\]
for $r$-query LDCs. As a result, the study of lower bounds for Matching Vector codes has gained importance. Based on the \cite{bdl13} upper bound, constant-query MV codes must satisfy
\[
N\ge K^{\Omega(\log\log K)}.
\]
for constant $m$.

\subsection{Our Results}

In this work, we show a substantially stronger subexponential upper bound for $3$-restricted MVFs
\begin{maintheorem}\label{thm:mvf-upper}
For every $n\ge 2$ and every integer $m$ satisfying $2\le m\le \sqrt n$,
\[
\MV(m,n,3)
\le
2^{O\left(\sqrt{n\log m\log n}\right)}.
\]
In particular, for constant $m$,
\[
\MV(m,n,3)
\le
2^{O\left(\sqrt{n\log n}\right)}.
\]
\end{maintheorem}

As a corollary, we can prove a strong lower bound on certain $3$-query matching vector codes. 

\begin{corollary}\label{cor:mvc-lower}
Every $3$-query matching vector code obtained from a $3$-restricted MVF
$\U,\V\subseteq \Z_m^n$ with $2\le m\le \sqrt n$, message length $K$, and
codeword length $N=m^n$, satisfies
\[
N
\ge
\exp\left(\Omega\left(\frac{(\log K)^2}{\log\log K}\right)\right).
\]
\end{corollary}

\subsection{Techniques} 

Our work continues the study of additive structure in $\U$ and $\V$. In \cite{bdl13}, this structure is explored through the lens of
approximate duality \cite{zewi}. In a different direction, \cite{aggarwal2025}
suggested studying collisions in sumsets of $\U$: many distinct subsets having
the same sum should impose strong structure on the matching-vector family. We
refer to the collection of all $s$-subsets whose corresponding vectors have the
same sum as a \emph{collision class}. Our proof develops this viewpoint directly
by showing that no collision class can be too large.

The key step is to show that a collision class cannot contain a large
moonflower\footnote{A family of sets is a moonflower if each set contains an
element that belongs to none of the others. Pairwise disjoint sets would also
suffice for our argument, but moonflowers capture the precise combinatorial
structure that we need.} \cite{lovett2026}. The polynomial method enters by
starting from the low-rank matrix of inner-product residues between $\U$ and
$\V$. From this matrix, we construct a new matrix that encodes the
inner-product structure between the equal sums and suitable vectors in $\V$.
The presence of a moonflower inside the collision class forces this new matrix
to have the form $J-I$, and hence to have large rank. At the same time, its
construction from the original low-rank residue matrix via a multilinear
polynomial gives a strong upper bound on its rank. Comparing these two bounds
limits the size of the moonflower.

Finally, let us fix a collision class and choose a maximal moonflower inside it. Because the moonflower is maximal, the union of its sets must intersect every
other set in the class. Otherwise, we could add a disjoint set and obtain a
larger moonflower. Thus every representation of the fixed $s$-sum contains at
least one element from this relatively small union. Once such an element is
fixed and removed, the remaining $s-1$ elements form a representation of a
corresponding $(s-1)$-sum. This bounds the size of an $s$-sum collision class in
terms of the size of the maximal moonflower and the largest $(s-1)$-sum
collision class. This recurrence gives a uniform bound on the size of
every $s$-sum collision class. Combining this with the fact that there are only
$m^n$ possible sums yields the desired upper bound on $K$.

\subsection{Concurrent Work}
Independently, Aggarwal and Obremski \cite{AggarwalObremski2026} obtain the same upper bound for $3$-restricted MVFs.
\subsection*{Acknowledgements} The author would like to thank Alexander Golovnev for helpful discussions during the writeup of this work. 
\subsection*{AI Usage} The main ideas behind this work were obtained while exploring multiple directions to understand the intuition developed in \cite{aggarwal2025} in extended conversations with ChatGPT 5.5. Specifically, the main idea behind the proof of \cref{lem:equal-sum-moonflower} was developed by ChatGPT 5.5 and refined into its current form by the author. Subsequent conversations were not able to provide any further insight into improving the bound for $3$-restricted MVFs  or for $r$-restricted MVFs more generally. While multiple LLMs were used in reviewing and editing the manuscript, all mathematical claims, proofs, and errors
are the sole responsibility of the author.
\section{Preliminaries}\label{sec:prelim}

% \subsection{Entropy}

% \begin{definition}[Min-Entropy]
%     Let $X$ be a random variable over a set $\mathcal{X}$, the Min-Entropy is defined as:
%     \[ \hinf(X) = \min_{x\in \mathcal{X}} -\log (\Pr[X=x]) \; . \]
% \end{definition}

% We will require the following simple fact about min entropy. 

% \begin{fact}\label{fact:entropy}
% For any random variable $X$ supported on $\mathcal{X}$,
% \[
% \hinf(X) \leq \log |\mathcal{X}|.
% \]
% \end{fact}

\subsection{Matching Vector Families}

We will utilize some known lemmas on restricted MV families. Succinctly, we can turn any $3$-restricted MVF modulo $m$ into a \emph{typical} $3$-restricted MVF whose inner products modulo $m$ satisfy some nice properties and $m$ itself is constrained in the following fashion. 

\begin{lemma}[\cite{bdl13, aggarwal2025}] \label{lem:residues_bdl}
Let $(\U, \V)$ be a $S$-matching vector family in $\Z_m^n$ and for any proper divisor $m' | m$, $(\U, \V)$ is not a matching vector family modulo $m'$. Then, $m$ has at most $|S|$ prime factors.
\end{lemma}

Let us now define a typical $3$-restricted MVF 

\begin{definition}[Typical $3$-restricted MVF]\label{def:typical-mvf}
We say a $3$-restricted MVF $\U, \V$ modulo $m$ is \emph{typical} if  $m=p^aq^b$, where $p,q$ are distinct primes and $a, b \geq 1$. Furthermore, the residues 
$S = \{\alpha,\beta\} \subseteq \Z_m\setminus\{0\}$ satisfy the following residue structure
\[
\alpha \equiv 0 \pmod{p^aq^{b-1}},
\qquad
\alpha \not\equiv 0 \pmod{q^b},
\]
and
\[
\beta \equiv 0 \pmod{p^{a-1}q^b},
\qquad
\beta \not\equiv 0 \pmod{p^a}.
\]
\end{definition}

Finally, the following lemma combined with the previous one allows us to convert a $3$-restricted MVF into a typical one without any loss of parameters. 

\begin{lemma}[\cite{aggarwal2025}] \label{lem:residues}
    Let $\U, \V$ be a $3$-restricted MVF modulo $m$ such that for any proper divisor $m' \mid m $, $\U, \V$ is not a $3$-restricted MVF modulo $m'$ and that $m$ is not a prime power. Then $\U, \V$ is a typical $3$-restricted MVF. 
\end{lemma}

\subsection{Moonflowers}

We will utilize the concept of moonflowers in our lower bound. Moonflowers were introduced by Lovett, Meka, and Wang in recent work inspired by sunflowers, with applications to code sparsification. We do not use any results from that work in our proof, but moonflowers are the right object to capture our lower bound structure.

\begin{definition}[Moonflower]\label{def:moonflower}
A family of sets $A_1,\ldots,A_L$ is an $L$-moonflower if for every $h\in[L]$, there is an element
\[
i_h\in A_h
\]
such that
\[
i_h\notin A_t
\qquad
\text{for every }t\ne h.
\]
We call $i_h$ a unique element of $A_h$.
\end{definition}

% \begin{lemma}[Moonflower-free families \cite{lovett2026}]\label{lem:moonflower_bound}
% There is an absolute constant $C$ such that the following holds. Let $\cA$ be a family of sets of size at most $w$ containing no $k$-moonflower. Then
% \[
% |\cA|
% \le
% \begin{cases}
% \left(Ck/w\right)^w, & w\le k,\\
% \left(Cw/k\right)^k, & w\ge k.
% \end{cases}
% \]
% \end{lemma}

\subsection{Linear Algebra}

In this section, we prove two auxiliary linear-algebraic lemmas. Our first lemma allows us to pass from matrices over the prime-power ring $\mathbb{Z}_{q^b}$ to matrices over the field $\mathbb{F}_q$, while preserving the rank bound needed in our main polynomial-method argument.
For $a \in q^{b-1}\mathbb Z_{q^b}$, let $\tau_q(a)\in\mathbb F_q$ denote the coefficient of $q^{b-1}$, that is,
\[
\tau_q(q^{b-1}c)=c\bmod q.
\]
This is well-defined.
\begin{lemma} \label{lem:coeff_rank}
Let $R=\mathbb Z_{q^b}$ where $q$ is a prime and $b \geq 1$. Let
\[
X\in R^{M\times n},
\qquad
Y\in R^{n\times N}.
\]
Suppose every entry of $XY$ lies in $q^{b-1}R$. Define $L\in\mathbb F_q^{M\times N}$ by
\[
L_{ij}=\tau_q((XY)_{ij}).
\]
Then
\[
\operatorname{rank}_{\mathbb F_q}(L)\le n.
\]
\end{lemma}

\begin{proof}
Take the Smith normal form of
\[
X=ADB
\]
over $R$, where $A,B$ are invertible and $D$ is diagonal with at most $n$ nonzero rows. Since $XY$ has entries in $q^{b-1}R$, so does
\[
DBY=A^{-1}XY.
\]
Let
\[
L'_{ij}=\tau_q((DBY)_{ij}).
\]
The matrix $L'$ has at most $n$ nonzero rows, hence
\[
\operatorname{rank}_{\mathbb F_q}(L')\le n.
\]
Since $XY=A(DBY)$ and $\tau_q(ra)=(r\bmod q)\tau_q(a)$ for $a\in q^{b-1}R$, we have
\[
L=\overline A L',
\]
where $\overline A := A \bmod q$ entrywise. Therefore
\[
\operatorname{rank}_{\mathbb F_q}(L)\le \operatorname{rank}_{\mathbb F_q}(L')\le n.
\]
\end{proof}

Our second lemma bounds the rank of a matrix obtained by applying a multilinear polynomial entrywise to several low-rank matrices. 

\begin{lemma} \label{lem:polynomial_rank}
Let $B_1,\ldots,B_d$ be $M\times N$ matrices over a field $\mathbb F$, and let
\[
\operatorname{rank}_{\mathbb F}(B_h)\le n
\qquad
\text{for every }h\in[d].
\]
Let $f\in\mathbb F[z_1,\ldots,z_d]$ be multilinear, and define $P$ entry-wise by
\[
P_{ij}=f((B_1)_{ij},\ldots,(B_d)_{ij}).
\]
Then
\[
\operatorname{rank}_{\mathbb F}(P)\le (n+1)^d.
\]
\end{lemma}

\begin{proof}
Write
\[
f(z_1,\ldots,z_d)=\sum_{T\subseteq[d]} c_T\prod_{h\in T}z_h.
\]
For each $T\subseteq[d]$, define $B_T$ entrywise by
\[
(B_T)_{ij}=\prod_{h\in T}(B_h)_{ij}.
\]
Thus $B_T$ is the Hadamard product of the matrices $B_h$ with $h\in T$. By the Hadamard-product rank inequality,
\[
\operatorname{rank}_{\mathbb F}(B_T)
\le
\prod_{h\in T}\operatorname{rank}_{\mathbb F}(B_h)
\le n^{|T|}.
\]
For $T=\emptyset$, $B_T$ is the all-ones matrix whose rank is $1$. Since
\[
P=\sum_{T\subseteq[d]} c_T B_T,
\]
rank subadditivity gives
\[
\operatorname{rank}_{\mathbb F}(P)
\le
\sum_{T\subseteq[d]} n^{|T|}
=
(n+1)^d.
\]
\end{proof}
\section{Proofs}

Throughout this section, let $\U = \{u_1, \dots, u_K\}$ and $\V = \{v_1, \dots, v_K\}$ be a typical $3$-restricted MVF modulo $m=p^aq^b$. For $A\subseteq [K]$, write
\[
u_A:=\sum_{i\in A}u_i.
\]
For $w\in\Z_m^n$ and $d\ge 1$, define
\[
C_w^{(d)}
=
\left\{A\in\binom{[K]}d:u_A=w\right\}.
\]

Our main lemma states that for each $w \in \Z_m^n$, there cannot be a large moonflower in $C_w^{(d)}$, the family of $d$-subsets whose corresponding matching vectors in $\U$ sum to $w$.

\begin{lemma}[Equal-sum moonflower bound]\label{lem:equal-sum-moonflower}
Let $\U,\V$ be a typical $3$-restricted MVF modulo $m=p^aq^b$. For every $w\in\Z_m^n$ and every $d\ge 1$, if
\[
A_1,\ldots,A_L \in C_w^{(d)}
\]
is a moonflower, then
\[
L\le d\left((n+1)^d+1\right).
\]
\end{lemma}

\begin{proof}
Fix $w\in\Z_m^n$, $d\ge 1$, and a moonflower
\[
A_1,\ldots,A_L \in C_w^{(d)}.
\]
The proof has three main components. We begin with the full matrix of inner products
\[
M_{ij}:=\langle u_j,v_i\rangle\in\Z_m .
\]
The diagonal entries of $M$ are $0$, while every off-diagonal entry lies in
$S = \{\alpha,\beta\}$.

\paragraph{Moving from $\Z_m$ to $\F_q$.} First, we reduce all inner products modulo $q^b$. Since
\[
\alpha\in q^{b-1}\mathbb Z_{q^b},
\qquad
\beta\equiv 0\pmod {q^b},
\]
and the diagonal entries are $0$, all entries lie in $q^{b-1}\mathbb Z_{q^b}$. Define
\[
Q_{ij}
=
\tau_q(\langle u_j,v_i\rangle).
\]
By \cref{lem:coeff_rank}, we know that
\[
\operatorname{rank}_{\mathbb F_q}(Q)\le n.
\]
Moreover, since $\alpha \not\equiv 0\pmod{q^b}$, there is a nonzero $\gamma\in\mathbb F_q$ such that
\[
Q_{ij}
=
\gamma\cdot \mathbf 1[\langle u_j,v_i\rangle=\alpha].
\]
Let $\Pi = \gamma^{-1} Q$, then we have $\operatorname{rank}_{\mathbb F_q}(\Pi) = \operatorname{rank}_{\mathbb F_q}(Q) \le n$. Note that $\Pi_{ij} = 1[\langle u_j,v_i\rangle=\alpha].$ $\Pi$ is a low-rank indicator matrix for the residue $\alpha$. 

\paragraph{From elements to sums.} Now, for each $h\in[L]$, choose a unique element
\[
i_h\in A_h,
\qquad
i_h\notin A_t\quad\text{for all }t\ne h. 
\]
We will use
$\Pi$ to build a new matrix indexed by the sets $A_1,\ldots,A_L$, capturing the
interaction between the sums $u_{A_t}$ and the unique element $v_{i_h}$.
Define
\[
\ell_h
=
\left|\{j\in A_h:\langle u_j,v_{i_h}\rangle=\alpha\}\right|.
\]
Since $i_h\in A_h$ gives the diagonal inner product $0$, we have
$\ell_h\in\{0,\ldots,d-1\}$. By the pigeonhole principle, some value $\ell$ occurs for at least $L/d$ of the sets. Restricting to those sets and relabeling, we obtain a sub-moonflower $A_1,\ldots,A_{L'}$ with $L'\ge L/d$ and $\ell_h=\ell$ throughout.

For $h\ne t$, let
\[
r_{ht}
=
\left|\{j\in A_t:\langle u_j,v_{i_h}\rangle=\alpha\}\right|.
\]
Since $i_h\notin A_t$, all terms in
$\langle u_{A_t},v_{i_h}\rangle$ are off-diagonal. Since
$u_{A_h}=u_{A_t}$, taking their respective inner products with $v_{i_h}$ gives
\[
r_{ht}\alpha+(d-r_{ht})\beta
\equiv
\ell\alpha+(d-1-\ell)\beta
\pmod m.
\]
Thus $r_{ht}$ must lie in the set
\[
T_\ell
=
\left\{
r\in\{0,\ldots,d\}:
r\alpha+(d-r)\beta
\equiv
\ell\alpha+(d-1-\ell)\beta
\pmod m
\right\}.
\]
Crucially, $\ell\notin T_\ell,$
since $r=\ell$ would force $\beta\equiv0\pmod m$. Therefore, along the diagonal block pair $(h,h)$ the row indexed by $i_h$
sees exactly $\ell$ occurrences of the residue $\alpha$, while along every
off-diagonal block pair $(h,t)$ it sees a number of occurrences belonging to
$T_\ell$, a set that avoids $\ell$.

\paragraph{Designing the polynomial method matrix.}  We can encode this with a function $g:\{0,1\}^d\to\mathbb F_q$ such that
\[
g(z)=0\quad\text{if }|z|=\ell,
\qquad
g(z)=1\quad\text{if }|z|\in T_\ell.
\] 
Let $f\in\mathbb F_q[z_1,\ldots,z_d]$ be the unique multilinear polynomial agreeing with $g$ on the Boolean cube. Fix orderings $A_t=\{j_{t,1},\ldots,j_{t,d}\}$ and define
\[
P_{ht}:=
f(\Pi_{i_h,j_{t,1}},\ldots,\Pi_{i_h,j_{t,d}}).
\]
By construction, the input has Hamming weight $\ell$ when $h=t$, and
Hamming weight in $T_\ell$ when $h\ne t$. Hence $P=J-I$, so
\begin{equation} \label{eqn:lower}
L/d - 1\leq L'-1\le \operatorname{rank}_{\mathbb F_q}(P).
\end{equation}
Now, for $r\in[d]$, define
\[
(B_r)_{ht}:=\Pi_{i_h,j_{t,r}}.
\]
Each $B_r$ is obtained from $\Pi$ by selecting rows and columns, possibly with
repetitions, so \\ $\operatorname{rank}_{\mathbb F_q}(B_r)\le n$. Since $P$ is
obtained by applying $f$ entrywise to $B_1,\ldots,B_d$, \cref{lem:polynomial_rank}
gives
\begin{equation} \label{eqn:upper}
\operatorname{rank}_{\mathbb F_q}(P)\le (n+1)^d.
\end{equation}

Combining \cref{eqn:lower} and \cref{eqn:upper} completes our proof.
\end{proof}

Using this lemma, we now complete the proof of the main theorem by an induction on the sumset size.

\begin{proof}[Proof of \cref{thm:mvf-upper}]
First, we may assume without loss of generality that $(\U,\V)$ does not remain a
$3$-restricted MVF modulo any proper divisor $m'\mid m$. Indeed, if it does,
replace $m$ by a divisor-minimal modulus with this property. If this modulus is
a prime power, then the known prime-power bound, together with $m\le\sqrt n$,
already gives the desired result. Otherwise, by
\cref{lem:residues_bdl,lem:residues}, we may assume that $(\U,\V)$ is typical.

Fix an integer $s\ge2$. For $1\le d\le s$, define
\[
F_d:=\max_{u\in\Z_m^n}|C_u^{(d)}|.
\]
We first show that $F_s\le n^{3s^2}$.
Let us fix $u$ and let $\mathcal F\subseteq C_u^{(s)}$ be a largest moonflower of size $M_s \leq 2s(n+1)^s$. Set \[
W:=\bigcup_{S\in\mathcal F}S.
\]
Then $|W|\le sM_s$. Moreover, every $T \in C_u^{(s)}$ intersects $W$, otherwise
$\mathcal F\cup\{T\}$ would be a larger moonflower. Hence
\[
\left |C_u^{(s)} \right|
\le
\sum_{i\in W} \left|\{T\in C_u^{(s)}:i\in T\}\right|
\le |W|F_{s-1}
\le sM_sF_{s-1}.
\]
Here the second inequality follows by mapping each set $T$ to a set in $C_{u-u_i}^{(s-1)}$ for some $i \in W \subseteq [K]$. 
Taking the maximum over $u$ gives
\[
F_s\le sM_sF_{s-1}.
\]
Plugging in \cref{lem:equal-sum-moonflower}, we obtain that
\[
F_s\le 2s^2(n+1)^sF_{s-1}.
\]
Since the $u_i$'s are distinct, $F_1=1$. Hence
\[
F_s\le \prod_{d=2}^s 2d^2(n+1)^d
\le 2^s(s!)^2(n+1)^{s(s+1)/2} \leq n^{3s^2}.
\]
when $s \geq 2$ and $n \geq 4$.
The sets $C_w^{(s)}$ partition $\binom{[K]}s$ as $w$ ranges over $\Z_m^n$, so
\[
\binom Ks
=
\sum_{w\in\Z_m^n}|C_w^{(s)}|
\le
m^nF_s
\le
m^n n^{3s^2}.
\]
Equivalently,
\[
\log\binom Ks\le n\log m+3s^2\log n.
\]
Now, set
\[
s=\left\lfloor \frac{\log K}{24\log n}\right\rfloor .
\]
If $s<2$, then $\log K=O(\log n)$. Otherwise, we have 
\begin{align*}
n\log m
&\ge \log\binom Ks-3s^2\log n \\
&\ge s\log(K/s)-3s^2\log n \\
&\ge c\frac{(\log K)^2}{\log n},
\end{align*}
for some absolute constant $c>0$, where the last inequality follows by plugging in $s$. Rearranging gives
\[
K\le 2^{O(\sqrt{n\log m\log n})}.
\]
\end{proof}

\bibliographystyle{alpha}
\bibliography{bibliography}
\end{document}